\documentclass[a4paper,12pt]{article}

\usepackage[latin1]{inputenc}
\usepackage[T1]{fontenc}
\usepackage{amsmath,amssymb}
\usepackage{color}
\usepackage{xspace}
\usepackage{newtimes}
\usepackage{graphicx}
\usepackage{amssymb,latexsym,amsmath,amsthm}
\usepackage{subfigure}
\usepackage{varioref}
\usepackage{indentfirst}
\usepackage{todonotes}
\usepackage[switch,displaymath, mathlines]{lineno}

\newcommand{\immagine}[3]{%
\begin{figure}[ht]%
\centering%
\includegraphics{#1}%
\caption{#2}%
\label{#3}%
\end{figure}%
}

\graphicspath{{immagini/}{immagini2/casoC=C'/}}

\newtheorem{Pro}{Proposition}
\newtheorem{Teo}{Theorem}
\newtheorem{Def}{Definition}

\newtheorem{Lem}{Lemma}

\newtheorem{Rem}{Remark}
\newtheorem{Not}{Notation}

\def\ZZ{{\mathbb{Z}}}
\def\RR{{\mathbb{R}}}
\def\CC{{\mathbb{C}}}

\newcommand{\apici}[1]{\lq\lq{#1}\rq\rq}

\begin{document}

\vspace{2cm}
\title{$0$-Gaps on $3$D Digital Curves}

\author{Angelo MAIMONE \and Giorgio NORDO\thanks{%
This research was supported by italian P.R.I.N., P.R.A. and I.N.D.A.M. (G.N.S.A.G.A.)}%
}

\date{}

\maketitle

\begin{abstract}
in Digital Geometry, gaps are some basic portion of a digital object that a discrete ray can cross
without intersecting any voxel of the object itself.
Such a notion is quite important in combinatorial image analysis and it is strictly connected with some applications in fields as CAD and Computer graphics.
In this paper we prove that the number of $0$-gaps of a $3$D digital curve can be expressed as a linear combination of the number of its $i$-cells (with $i = 0, \ldots, 3$).
\end{abstract}

\section{Introduction}
With the word  {\lq\lq{gap}\rq\rq} in Digital Geometry we mean some basic portion of a digital object that
a discrete ray can cross without intersecting any voxel of the object itself.
Since such a notion is strictly connected with some applications in the field of Computer graphics
(e.g. the rendering of a 3D image by the ray-tracing technique), many papers (see for example \cite{gapscongresso}, \cite{gaps},  \cite{gapsdue}, and \cite{genus}) concerned the study of $0$- and $1$-gaps of $3$-dimensional objects.

More recently, in \cite{mn2011} and \cite{mn2013} two formulas which express, respectively the number of
$1$-gaps of a generic 3D object of dimension $\alpha=1,2$ and
the number of $(n - 2)$-gaps of a generic digital $n$-object,
by means of a few simple intrinsic parameters of the object itself were found.
Furtermore, in \cite{mn2015} the relationship existing between the dimension of a $2$D digital object
equipped with an adjacency relation $A_\alpha$ ($\alpha \in \{0,1\}$) and the number of its gaps was investigated.

In the next section we recall and formalize some basic definitions and properties of the general $n$-dimensional digital spaces with particular regard to the notions
of block, tandem and gap.

In Section \ref{sez:main}, we restrict our attention to digital curves in $3$D digital spaces, deriving some particular cases of the propositions above recalled
in order to prove our main result which states that
the number $g_0$ of $0$-gaps of a $3$D digital curve $\gamma$ can be expressed as a linear combination of the number $c-i$ of its $i$-cells (with $i = 0, \ldots, 3$)
and more precisely that $g_0 = \sum_{i = 0}^3 (-1)^{i+1} 2^i c_i$.

\section{Preliminaries}\label{sez:preliminari}

Throughout this paper we use the \textit{grid cell model} for
representing digital objects, and we adopt the terminology from
\cite{klette-rosenfeld} and \cite{kovalevsky}.
\par
Let $x=(x_1,\ldots x_n)$ be a point of $\ZZ^n$, $\theta \in
\{-1,0,1\}^n$ be an $n$-word over the alphabet $\{-1,0,1\}$, and
$i\in\{1,\ldots n\}$. We define $i$-cell related to $x$ and
$\theta$, and we denote it by $ e = (x,\theta)$, the Cartesian
product, in a certain fixed order, of $n-i$ singletons $\left\{ x_j
\pm \frac{1}{2} \right\}$ by $i$ closed sets $\left[x_j -
\frac{1}{2}, x_j + \frac{1}{2}\right]$, i.e. we set
\[
e = (x,\theta) = \prod_{j=1}^n \left[ x_j  + \frac 1 2 \theta_j - \frac 1 2 [\theta_j=0], x_j  + \frac 1 2 \theta_j + \frac 1 2 [\theta_j=0]  \right],
\]
where $[\bullet]$ denotes the Iverson bracket \cite{Knuth}. The
word $\theta$ is called the \textit{direction} of the cell
$(x,\theta)$ related to the point $x$.
\\
Let us note that an $i$-cell can be related to different point $x
\in \ZZ^n$, and, once we have fixed it, can be related to different
direction. So, when we talk generically about $i$-cell, we mean one
of its possible representation.

The dimension of a cell $e = (x, \theta)$, denoted by  $\dim(e)=i$,  is the number of non-trivial interval of its product representation, i.e. the number of null components of its direction $\theta$.
Thus, $\dim(e) = \sum_{j = 1}^ n [\theta_j = 0]$ or, equivalently, $\dim(e)= n - \theta \cdot \theta$.
So, $e$ is an $i$-cell if and only if it  has dimension $i$.

We denote by $\CC_n^{(i)}$ the set of all $i$-cells of $\RR^n$ and
by $\CC_n$ the set of all cells defined in $\RR^n$, i.e. we set
$\CC_n=\bigcup_{j=0}^n \CC_n^{(j)}$. An $n$-cell of $\CC_n$ is also
called an $n$-voxel. So, for convenience, an $n$-voxel is denoted by
$v$, while we use other lower case letter (usually $e$) to denote
cells of lower dimension. A finite collection $D$ of $n$-voxels is a
digital $n$-object. For any $i=0,\ldots, n$, we denote by $C_{i}(D)$
the set of all $i$-cells of the object $D$, that is $D \cap
\CC_n^{(i)}$, and by $c_i(D)$ (or simply by $c_i$ if no confusion
arise) its cardinality $|C_i(D)|$.

We say that two $n$-cells $v_1$, $v_2$ are $i$-adjacent ($ i =
0,1,\ldots, n-1$) if $ v_1\,\ne\, v_2 $ and there exists at least an
$i$-cell $\overline{e}$  such that $ \overline{e} \subseteq v_1 \cap
v_2$, that is if they are distinct and share at least an $i$-cell.
Two $n$-cells $v_1$, $v_2$ are \textit{strictly} $ i $-adjacent, if
they are $i$-adjacent but not $j$-adjacent, for any $j>i$, that is
if $v_1\cap v_2 \in \CC_n^{(i)}$. The set of all $n$-cells that are
$i$-adjacent to a given $n$-voxel $v$ is denoted by $A_i(v)$ and
called the $i$-\textit{adjacent neighborhoods} of $v$. Two cells
$v_1,v_2 \in \CC_n$ are \textit{incident} each other,  and we write $e_1 I e_2$, if $e_1
\subseteq e_2$ or $e_2 \subseteq e_1$.

\begin{Def}\label{Def: relazione di limitatezza}
Let $e_1,e_2 \in \CC_n$. We say that $e_1$ bounds $e_2$ (or that
$e_2$ is bounded by $e_1$), and we write $e_1 < e_2$, if $e_1 I e_2$
and $\dim(e_1) < \dim(e_2)$. The relation $<$ is called bounding
relation.
\end{Def}


\begin{Def}
An incidence structure (see \cite{Design_Theory})  is a triple $(V,\mathcal{B}, \mathcal I)$
where $V$ and $\mathcal{B}$ are any two disjoint sets and $\mathcal
I$ is a binary relation between $V$ and $\mathcal B$, that is
$\mathcal I \subseteq V \times \mathcal B$. The elements of $V$ are
called points, those of $\mathcal{B}$ blocks. Instead of $(p, B) \in
\mathcal I$, we simply write $p \mathcal I B$ and say that
\apici{the point $p$ lies on the block $B$} or \apici{$p$ and $B$
are incident}.
\end{Def}

If $p$ is any point of $V$, we denote by $(p)$ the set of all blocks
incident to $p$, i.e. $(p)=\{B\in \mathcal{B} \colon p \mathcal I
B\}$. Similarly, if $B$ is any block of $\mathcal B$,  we denote by
$(B)$ the set of all points incident to $B$, i.e. $(B)=\{p \in V
\colon p \mathcal I B \}$. For a point $p$, the number $r_p = |(p)|$
is called the degree of $p$, and similarly, for a block $B$, $k_B =
|(B)|$ is the degree of $B$.

Let us remind the following fundamental proposition of incidence structures.

\begin{Pro}\label{Pro:relazione fondamentale delle strutture di incidenza}
Let $(V,\mathcal{B}, \mathcal I)$ be an incidence structure. We have
\begin{equation}\label{Eq: Equazione fondamentale delle strutture di incidenza}
    \sum_{p \in V} r_p = \sum_{B \in \mathcal{B}} k_B,
\end{equation}
where $r_p$ and $k_B$ are the degrees of any point $p \in V$ and any
block $B \in \mathcal B$, respectively.
\end{Pro}

\begin{Def}\label{Def:blocchi}
Let $e$ be an $i$-cell (with $0\le i \le n-1$) of $\CC_n$.
Then an $i$-block centered on $e$, denoted with $B_i(e)$, is the union of all the $n$-voxels
bounded by $e$, i.e. $B_i(e)= \bigcup\{ v \in \CC_{n}^{(n)} \colon e
< v \}$.
\end{Def}

\begin{Rem}\label{Not:Numero di nVoxel che compone un blocco}
Let us note that, for any $i$-cell $e$, $B_i(e)$ is the union of
exactly $2^{n-i}$ $n$-voxels and $e \in B_i(e)$.
\end{Rem}

\begin{Def}\label{Def:tandemn}
Let $v_1$, $v_2$ be two $n$-voxels of a digital object $D$, and $e$
be an $i$-cell ($i=0,\ldots, n-1$). We say that $\{v_1,v_2\}$
forms an $i$-tandem of $D$ over $e$ and we will denote it by $t_i(e)$, if $D \cap B_i(e)= \{v_1,v_2\}
$, $v_1$ and $v_2$ are strictly $i$-adjacent and $v_1 \cap v_2 = e$.
\end{Def}

\begin{Def}\label{Def:gapn}
Let $D$ be a digital $n$-object and $e$ be an $i$-cell (with $i =
0,\ldots,n-2 $). We say that $D$ has an $i$-gap over $e$ if there
exists an $ i $-block $B_i(e)$ such that $B_i(e) \setminus D$ is an
$i$-tandem over $e$. The cell $e$ is called $i$-hub of the related
$i$-gap. Moreover, we denote by $g_i(D)$ (or simply by $g_i$ if no
confusion arises) the number of $i$-gap of $D$.
\end{Def}
\immagine{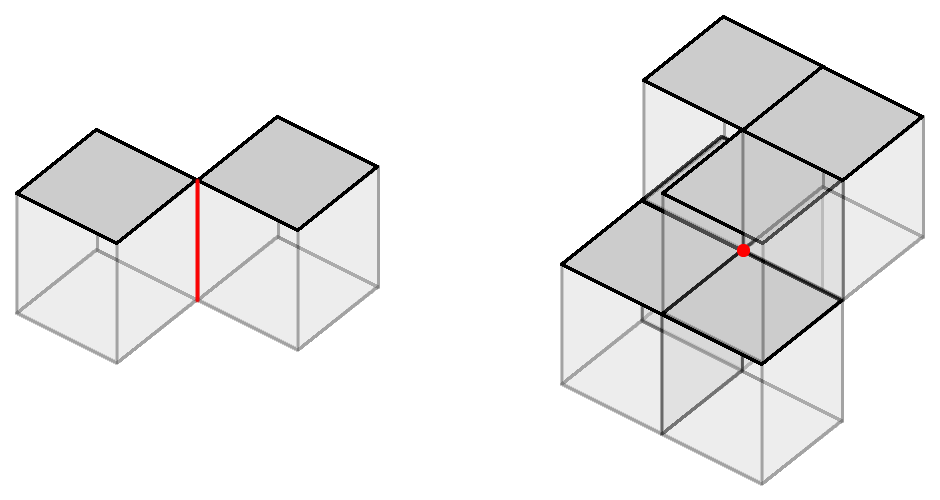}{Configurations of  $1$- and $0$-gaps in $\CC_3$.}{gap}

\begin{Not}
For any $i=0, \ldots, n-1$, we denote by $\mathcal H_{i}(D)$ (or simply by $\mathcal H_{i}$ if no confusion arises)
the sets of all $i$-hubs of $D$.
Clearly, we have $|\mathcal H_{i}| = g_{i}$.
\end{Not}

\begin{Def}\label{Def:FreeCell}
An $i$-cell $e$  (with $i=0, \ldots, n-1$) of a digital $n$-object $D$ is
free iff $B_i(e) \nsubseteq D$.
\end{Def}

\begin{Not}
For any $i=0,\ldots,n-1$, we denote by $C_i^*(D)$ (respectively by
$C_i'(D)$) the set of all free (respectively non-free) $i$-cells of
the object $D$. Moreover, we denote by $c_i^*(D)$ (or simply by
$c_i^*$) the number of free $i$-cells of $D$, and by $c_i'(D)$ (or
simply by $c_i'$) the number of non-free cells.
\end{Not}

\begin{Rem}
It is evident that
$\{C_i^*(D),C_i'(D)  \}$ forms a partition of $C_i(D)$ and that
$c_i= c_i^* + c_i'$.
\end{Rem}

\begin{Pro}\label{Pro:c2}
Let $D$ be a digital $n$-object. Then
$$c_{2} = 6 c_3 - c_{2}'. $$
\end{Pro}
\begin{proof}
Let us consider the set
\[ F = \bigcup_{v \in   C_n(D)} \{ (e,v) \colon e \in C_{n-1}(D), e < v \}.\]
It is evident that:
\begin{align*}
\big|F \big| & =  \Big|\{ (e,v) \colon  e \in C_{n-1} (D), e < v \} \Big|   \cdot  \Big| C_{n}(D)\Big| \\
             & = c_{n-1 \rightarrow n} \cdot c_n \\
             & = 2 n c_n . \\
\end{align*}
Let us set:
$$F^* = F \cap (C_{n-1}^*(D) \times C_{n}(D))$$
and
$$F' = F \cap (C_{n-1}'(D) \times C_{n}(D)) .$$
The map $\phi \colon F^* \to C_{n-1}^*(D)$, defined by $\phi(e,v)=e$, is a bijection.
In fact, besides being evidently surjective, it is also injective, since, if by contradiction
there were two distinct pairs $(e,v_1)$ and $(e,v_2) \in F^*$ associated to $e$,
then $B_{n-1}(e)= \{v_1,v_2\}$ should be an $(n-1)$-block contained in $D$.
This contradicts the fact that the $(n-1)$-cell $e$ is free. Thus $|F^*| = |C_{n-1}^*(D)| = c_{n-1}^*$.
\\
On the other hand, it results:
\begin{align*}
\big|F' \big| & = \Big| \displaystyle \bigcup_{v \in   C_n (D)} \{ (e,v) \colon e \in C'_{n-1}(D), e < v \} \Big| \\
               & = \Big|\displaystyle \bigcup_{e \in   C'_{n-1} (D)} \{ (e,v) \colon v \in C_{n}(D), e < v \} \Big| \\
               & = \Big|\{ (e,v) \colon  v \in C_{n} (D), e < v \} \Big| \cdot  \Big| C'_{n-1}(D)\Big| \\
               & = c_{n-1 \leftarrow n} \cdot c'_{n-1} \\
               & = 2 c'_{n-1}.
\end{align*}
Since $\{ F^*, F' \}$ is
a partition of $F$, we finally have that $| F | = |F^*| + |F'|$,
that is $   2n c_n = c_{n-1}^* + 2 c_{n-1}' = c_{n-1} -  c_{n-1}' +
2  c_{n-1}' = c_{n-1} + c_{n-1}'$, and then the thesis.
\end{proof}

\begin{Not}
Let $i,j$ be two natural number such that $0 \le i < j$. We denote
by $ c_{i \rightarrow j} $ the maximum number of $i$-cells of
$\CC_n$ that bound a $j$-cell. Moreover, we denote by $c_{i
\leftarrow j}$ the maximum number of $j$-cell of $\CC_n$ that are
bounded by an $i$-cell.
\end{Not}

The following three propositions were proved in \cite{mn2013}

\begin{Pro}\label{Pro:c_i_verso_j}
For any $i,j \in \mathbb{N}$ such that $0 \le i <j$, it is
\[
c_{i \rightarrow j} = 2^{j-i} \binom{j}{i}.
\]
\end{Pro}

\begin{Pro}\label{Pro:c_i_da_j}
For any $i,j \in \mathbb{N}$ such that $0 \le i <j$, it is
\[
c_{i \leftarrow j} = 2^{j-i}\binom{n-i}{j-i}.
\]
\end{Pro}

\begin{Pro}
Let $D$ be a digital $n$-object. Then
\[ c_{n-1} = 2n c_n - c_{n-1}'. \]
\end{Pro}

\begin{Not}
Let $e$ be an $i$-cell of a digital $n$-object $D$, and $0 \le i <
j$. We denote by $b_j(e,D)$ (or simply by $b_j(e)$ if no confusion
arises) the number of $j$-cells of $bd(D)$ that are bounded by $e$.
\end{Not}

Let us note that if $e$ is a non-free $i$-cell, then $b_j(e) = 0$.

\begin{Pro}
Let $v$ be an $n$-voxel and $e$ be one of its $i$-cells, $i=0,\ldots, n-1$.
Then, for any $i < j \le n$, it results:
\[b_j(e) = \frac{c_{i \rightarrow j}c_{j \rightarrow n}}{c_{i \rightarrow n}} .\]
\end{Pro}


\section{Gaps and curves in $3$D digital space}\label{sez:main}

Throughout the rest of the paper we will consider the $3$-dimensional digital space $\ZZ^3$ with the corresponding grid cell model $\CC_3$.

\begin{Def}
A digital object $\gamma$ of $\CC_3$ is said a \textit{digital $k$-curve} if it satisfies the following two condition:
\begin{itemize}
  \item $\forall v \in \gamma$ it is $ 1 \le|A_k(v)|\le 2$;
  \item  For any $v \in \gamma$, if $v_1, v_2 \in A_k(v)$, then $\{v_1, v_2\} \not \in A_k(v) $,
\end{itemize}
that is, for any voxel $v \in \gamma$ there exist at most two voxels $k$-adjacent to $v$
and every pair of voxels $k$-adjacent to a voxel of $\gamma$ can not be $k$-adjacent to each other.
\end{Def}

The voxels in $\gamma$ which have only one $k$-adjacent voxel are said the extreme points of the curve.

\immagine{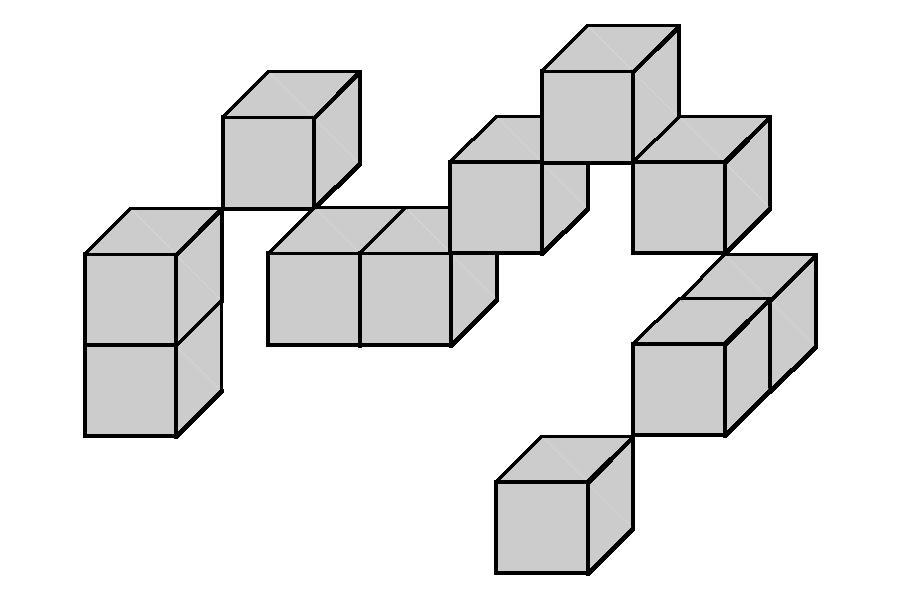}{An example of digital $0$-curve in $\CC_3$}{Excurva}

We are interested only to digital $0$-curve and, if no confusion arises, we will briefly call them digital curve.

The following propositions derive from some general ones proved in \cite{mn2013} for the $n$-dimensional case.

\begin{Pro}\label{Pro:3}
Let $v$ be a voxel and $e$ be one of its $i$-cell, $i=0,\ldots,
2$. Then, for any $i < j \le 2$, we have
\[ b_j(e) = \binom{3-i}{j-i} . \]
\end{Pro}

\begin{Pro}\label{Pro:1b}
Let $e$ be a $2$-cell of $\CC_3$. Then the number of $i$-cells ($i = 0,\ldots,2$)
of the $2$-block centered on $e$ is
\[ c_i(B_{2}(e)) =  \frac {9+i}{6} c_{i \rightarrow 3}.\]
\end{Pro}



In order to obtain our main result, we preliminarily need to prove the following result.

\begin{Pro}\label{Pro:E}
The number of $i$-cells ($i = 0, 1$) of an $1$-tandem $t_{1}(e)$ is
\[ c_i(t_{1}) = \frac{42 + 5i - i^2}{24} c_{i \to 3}.\]
\end{Pro}
\begin{proof}
By Definition, $t_{1}(e)$ is composed of two strictly
$1$-adjacent voxels. Each of such voxels has exactly
$c_{i \to 3}$ $i$-cells. But some of these cells are repeated onto
$t_1(e)$. The number of these repeated $i$-cells coincides with the
number of $i$-cells of the $1$-hub $e$. Since
\[ \binom{1}{i} = \frac{(n-i)(n-i-1)}{n(n-1)} \binom{n}{i}, \]
we have:
\begin{align*}
    c_i(t_{n-2}(e)) & = 2 c_{i \to n} - c_{i \to n-2} \\
                 & = 2 \cdot 2^{n-i} \binom{n}{i} -
                 2^{n-2-i}\binom{n-2}{i}\\
                 & =\frac{7 n^2 - 7 n + 2 i n - i^2 - i}{4 n (n - 1)} c_{i \to
                 n}.\\
\end{align*}
\end{proof}

The following useful proposition was proved in \cite{mn2011}.

\begin{Pro}\label{Pro:number of G1 gaps}
The number of $1$-gaps of a digital object $D$ of $\CC_3$ is given by:
\begin{equation}\label{Eq:mia}
    g_{1} = 2 c_{2}^*-c_{1}^*.
\end{equation}
\end{Pro}

\begin{Pro}\label{Pro:6}
Let $e$ be a free vertex that bounds the center $e'$ of a $2$-block
$B_2(e')$. Then $b_1(e) = 4$.
\end{Pro}

\begin{proof}
Let us consider the incidence structure $(C_0(B_2(e')),
C_1(B_2(e')), <)$. By Proposition \ref{Pro:relazione fondamentale
delle strutture di incidenza}, we have
\[ \sum_{a \in C_0(B_2(e'))} r_a = \sum_{a \in C_1(B_2(e'))}k_a. \]
Let us note that, by Proposition \ref{Pro:1b}, we have
$|C_1(B_2(e'))| = 20$ and $|C_0(B_2(e'))| = 12$.
\\
Since, for any $a \in C_1(B_2(e'))$ it is $k_a = c_{0 \to 1} = 2$,
we have
\begin{equation}\label{Equ:1}
    \sum_{a \in C_1(B_2(e'))} k_a = 2 \cdot | C_1(B_2(e'))| = 40.
\end{equation}
Let us now consider the sets:
\[ F = \{ a \in C_0(B_2(e')) \colon a < e' \} \]
and
\[ G = \{a \in C_0(B_2(e')) \colon a \nless e'\}. \]
Since $\{ F,G \}$ forms a partition of $C_0(B_2(e'))$, we can write
\[ \sum_{a \in C_0(B_2(e'))} r_a = \sum_{a \in F} r_a + \sum_{a \in G} r_a. \]
For any $a \in F$, let us set $r_a = b_1(e)$. We have
\begin{equation}\label{Equ:2}
    \sum_{a \in F} r_a  = |F| b_1(e) = c_{0 \to 2} b_1(e) = 4 b_1(e).
\end{equation}
Instead, thanks to Proposition \ref{Pro:3}, for any $a \in G$, it
is
\[ r_a = b_1(a) = \binom{3-0}{1-0} =3, \]
and so
\begin{equation}\label{Equ:3}
    \sum_{a \in G} r_a  = 3 \cdot |G| = 3 (|C_0(B_2(e'))| - c_{0 \to 2}) = 3(12 - 4) = 24.
\end{equation}
To sum up, by using Equations \eqref{Equ:1}, \eqref{Equ:2}, and
\eqref{Equ:3},we can write $4 b_1(e) + 24 = 40$, from which we get
the thesis.
\end{proof}

\begin{Pro}\label{Pro:A}
Let $\gamma$ be a digital curve of $\CC_3$. Then the number of $0$-cells that
bound some non-free $2$-cell is $4 c_2'$.
\end{Pro}
\begin{proof}
Since $c'_2(\gamma)$ coincides with the number of $2$-block of
$\gamma$, and since any non-free $2$-cell is bounded by $c_{0 \to 2}
= 4$ $0$-cells, the number of $0$-cells that bound some non-free
$2$-cell il exactly $4 c_2'$.
\end{proof}

\begin{Pro}\label{Pro:c_i_da_j}
For any $i,j \in \mathbb{N}$ such that $0 \le i <j$, it is
\[
c_{i \leftarrow j} = 2^{j-i}\binom{n-i}{j-i}.
\]
\end{Pro}

\begin{Pro}\label{Pro:C}
Let $D$ be a digital object
of $\CC_3$ and $e \in \mathcal H_0$. Then $b_1(e)=6$.
\end{Pro}
\begin{proof}
Since the number $b_1(e)$ of $1$-cells of $D$ bounded by
$e$ coincides with the maximum number of $1$-cells bounded by a
$0$-cell, that is, by Proposition \ref{Pro:c_i_da_j}
\[ b_{1}(e) = c_{0 \leftarrow 1} = 2^{1-0} \binom{3-0}{1-0} =6. \]
\end{proof}

We have the following lemma.
\begin{Lem}\label{Lem:E}
The number of  $0$-cells and $1$-cells of a $1$-tandem $t_1(e)$ is $c_0(t_1(e))=14$ and
$c_1(t_1(e)) = 23$, respectively.
\end{Lem}
\begin{proof}
It directly follows by Proposition \ref{Pro:E} for $n=3$ and $i=0$ or  $i=1$, respectively.
\end{proof}

\begin{Pro}\label{Pro:F}
Let $e$ be a $0$-cell that bounds a $1$-hub. Then $b_1(e) = 5$.
\end{Pro}
\begin{proof}
Let $e'$ a $1$-hub that is bounded by $e$, and $t_1(e')$ the related $1$-tandem.
Moreover, let us consider the incidence structure $(C_0(t_1(e')),
C_1(t_1(e')), <)$. By Proposition \ref{Pro:relazione
fondamentale delle strutture di incidenza}, we can write
\[ \sum_{a \in C_0(t_1(e'))} r_a = \sum_{a \in C_1(t_1(e'))} k_a. \]
By Lemma \ref{Lem:E}, we have $|C_0(t_1(e'))| = 14$ and
 $|C_1(t_1(e'))|= 23$. Moreover, since for any $ a \in C_1(t_1(e'))$, $k_a = c_{0 \to
 2}=2$, it is
 \[ \sum_{a \in C_1(t_1(e'))} k_a  = 2 \cdot |C_1(t_1(e'))| = 46.\]
 Now, let us set
 \[ F = \{ a \in C_0(t_1(e')) \colon a < e' \}  \]
and
 \[ G= \{  a \in C_0(t_1(e')) \colon a \nless e' \}. \]
 Since $\{ F,G \}$ is a partition of $C_0(t_1(e'))$, we have
 \[ \sum_{a \in C_0(t_1(e'))} r_a = \sum_{a \in F} r_a + \sum_{a \in G} r_a. \]
 Let us calculate $\sum_{a \in F} r_a$. If we set $r_a = b_1(e)$, we
 have
 \[ \sum_{a \in F} r_a = |F| b_1(e) = c_{o \to 1} b_1(e) = 2 b_1(e).\]
 Now, let us calculate $\sum_{a \in G} r_a$. By Proposition
 \ref{Pro:3}, for any $a \in G$, it is
 \[ r_a = b_1(a) = 3.  \]
Hence we get
\[ \sum_{a \in G} r_a = 3 \cdot |G| = 3 (|C_0(t_1(e'))| - c_{0 \to 1}) = 36.\]
To sum up, we have $2 b_1(e) + 36 = 46$, from which we get $b_1(e) =
5$.
\end{proof}

\begin{Pro}\label{Pro:Fb}
Let $\gamma$ be a digital curve of $\CC_3$. Then the number of $0$-cells that
bounds some $1$-hub of $\gamma$ is $2g_1$.
\end{Pro}
\begin{proof}
Since any $1$-hub is bounded by $c_{0\to
1}$ $0$-cell, we have that the number of $0$-cells that bound some
$1$-hub is exactly $2 g_1$.
\end{proof}

By applying  Proposition \ref{Pro:3} with  i = 0 and $j = 1$ we can easily prove the following proposition.

\begin{Pro}\label{Pro:G}
Let $e$ be a $0$-cell of a voxel $v \in \CC_3$. Then $b_1(e) = 3$.
\end{Pro}

\begin{Teo}\label{Teo:main}
Let $\gamma$ be a digital curve of $\CC_3$. Then the number of its $0$-gaps is given by:
\[g_0 = \sum_{i = 0}^3 (-1)^{i+1} 2^i c_i. \]
\end{Teo}
\begin{proof}
Let us consider the incidence structure $(C_0(\gamma),
C_1(\gamma),<)$. By Preposition \ref{Pro:relazione fondamentale
delle strutture di incidenza}, it is
\[ \sum_{a \in C_0(\gamma)} r_a = \sum_{a \in C_1(\gamma)} k_a. \]
Evidently, for any $a \in C_1(\gamma)$, we have that $k_a = 2$. So
\begin{equation}\label{Equ:principale1}
    \sum_{a \in C_1(\gamma)} k_a = 2 \cdot |C_1(\gamma)| = 2 c_1.
\end{equation}


Let us denote by $H_i(\gamma)$, $i=0,1$, and by $C'_2(\gamma)$,
the sets of $0$- and $1$-hubs and the set of non-free $2$-cells of
$\gamma$, respectively.
\\
Let us now calculate $\sum_{a \in \CC_0(\gamma)} r_a$. In order to
do that, let us consider the following sets of $0$-cells.
\\
\[A = \{ c \in \CC_0(\gamma) \colon c \in H_0(\gamma) \}.\]
\[B = \{ c \in \CC_0(\gamma) \colon c<e, e \in H_1(\gamma)  \}.\]
\[C = \{ c \in \CC_0(\gamma) \colon c <e, e \in \CC'_2(\gamma). \}\]
\[D = \CC_0(\gamma) \setminus (A \cap B \cap C). \]
Since $\{A,B,C,D\}$ forms a partition of $\CC_0(\gamma)$, we have
\[ \sum_{a \in \CC_0(\gamma)} r_a = \sum_{a \in A} r_a + \sum_{a \in B} r_a + \sum_{a
\in C} r_a + \sum_{a \in D} r_a.\]

Let us calculate $\sum_{a \in A} r_a$. By Proposition \ref{Pro:C},
for any $a \in A$ it is $r_a = 6$. Evidently $|A| = g_0$. Hence
\begin{equation}\label{Equ:a}
    \sum_{a \in A} r_a = r_a \cdot |A| =6 g_0.
\end{equation}
Let us calculate $\sum_{a \in B} r_a$. By Proposition \ref{Pro:F},
for any $a \in B$, it is $r_a = 5$. Moreover, by Proposition
\ref{Pro:Fb}, it is $|B|=2 g_1$. So
\begin{equation}\label{Equ:b}
    \sum_{a \in A} r_a = r_a \cdot |B| = 10 g_1.
\end{equation}
Let us calculate $\sum_{a \in C} r_a$. By Proposition \ref{Pro:6},
for any $a \in C$,  $r_a = 4$, and, by Proposition \ref{Pro:A}, $|C|
= 4 c_2'$. It follows that
\begin{equation}\label{Equ:c}
    \sum_{a \in A} r_a = r_a \cdot |C| = 16 c'_2.
\end{equation}
Finally, let us calculate $\sum_{a \in D} r_a$. By Proposition
\ref{Pro:G}, for any $a \in D$, it is $r_a = 3$. Moreover, $|D| =
c_0 - 4 c_2' - 2g_1 - g_0$. So
\begin{equation}\label{Equ:d}
    \sum_{a \in A} r_a = r_a \cdot |D| =  3 (c_0 - g_0 - 2 g_1 - 4
    c'_2).
\end{equation}
Combining the Equations \eqref{Equ:a},\eqref{Equ:b},\eqref{Equ:c},
and \eqref{Equ:d} we obtain $ 6 g_0 + 10 g_1 + 16 c'_2 + 3c_0 - 3
g_0 - 6 g_1 - 12 c'_2 = 2 c_1$, that is
\begin{equation}\label{Equ:CSI}
    3 c_0 + 4 c_2' + 4 g_1 + 3 g_0 = 2 c_1.
\end{equation}

Using Proposition \ref{Pro:number of G1 gaps}, we get $3 c_0 + 4c_2' + 8c_2^* - 4 c_1 + 3g_0 = 2c_1$, that is, since $c_2 =c_2' +c_2^*$,
$3 c_0 + 4c_2 + 4c_2^* + 3 g_0 = 6c_1$.
Moreover, by Proposition \ref{Pro:c2}, we get $-c_2' =c_2 - 6c_3$.
So we can write
\begin{align*}
   c_2^*  & =c_2 -c_2'\\
           & =c_2 +c_2 - 6c_3 \\
           &  = 2c_2 - 6c_3.
\end{align*}
Substituting the last expression in Equation \eqref{Equ:CSI}, we have
\[ 3 c_0 + 4c_2 + 8c_2 -24c_2 + 3 g_0 = 6c_1, \]
that is
\[ 3 c_0 + 12c_2 - 24 c_3 + 3 g_0 = 6c_1, \]
from which we finally  get
\[ g_0 = \sum_{i = 0} ^3 (-1)^{i + 1} 2^i  c_i. \]
\end{proof}


{\noindent {\it Key words and phrases:}}
digital geometry, digital curve, $0$-gap, $i$-tandem, $i$-hub, adjacency relation,
grid cell model, free cell.

\vskip 2pt

{\noindent {\it AMS Subject Classification:}} Primary: 52C35;
Secondary: 52C99. \vskip 15pt

\parskip=0pt {\noindent {\sc Giorgio NORDO \newline
Dipartimento di Matematica e Informatica, Universit\`a degli Studi di Messina, \newline
Viale F. Stagno D'Alcontres, 31 --
Contrada Papardo, salita Sperone, 31 - 98166 Sant'Agata -- Messina
(ITALY)}} \vskip 1pt \noindent
E-mail:  {\tt giorgio.nordo@unime.it}
\vskip 4mm
{\noindent {\sc Angelo MAIMONE \newline
E-mail:  {\tt angelomaimone@libero.it}

\end{document}